\documentclass{article}
\usepackage{geometry}
\usepackage[utf8]{inputenc}
\usepackage{amssymb}
\usepackage{amsmath}
\usepackage{amsthm}
\usepackage{amsfonts}
\usepackage{color}
\usepackage{tikz}
\usepackage{pgf}
\usetikzlibrary{shapes,arrows,positioning,calc} 
\usepackage{tabu}
\usepackage{hyperref}
\usepackage{longtable}
\usepackage{array}
\usepackage{siunitx,multirow,tabularx,booktabs}
\usepackage{amsmath,amssymb,bm,mathtools}

\usepackage{algorithm}
\usepackage[noend]{algpseudocode}

\makeatletter
\def\BState{\State\hskip-\ALG@thistlm}
\makeatother

\usepackage{thmtools}
\usepackage{thm-restate}
\newtheorem{theorem}{Theorem}

\newtheorem{proposition}[theorem]{Proposition}

\newtheorem{example}[theorem]{Example}

\begin{document}

\title{Transcendence of Sturmian Numbers over an Algebraic Base}

\author{
{\sc Florian~Luca}\\
{School of Mathematics, University of the Witwatersrand}\\
{Private Bag 3, Wits 2050, South Africa}\\
{Research Group in Algebraic Structures and Applications}\\
{King Abdulaziz University, Jeddah, Saudi Arabia}\\
%{Max-Planck Institute for Software Systems, Saarbr\"ucken, Germany}\\
{florian.luca@wits.ac.za}
\and
{\sc Jo\"el Ouaknine}\\
{Max Planck Institute for Software Systems}\\
{Saarland Informatics Campus, Saarbr\"ucken, Germany  }\\
{joel@mpi-sws.org}
\and
{\sc James Worrell}\\
{Department of Computer Science}\\
{University of Oxford, Oxford OX1 3QD, UK}\\
{jbw@cs.ox.ac.uk}}

\date{}
\maketitle

\begin{abstract}
  We consider numbers of the form
  $S_\beta(\boldsymbol{u}):=\sum_{n=0}^\infty \frac{u_n}{\beta^n}$ for
  $\boldsymbol{u}=\langle u_n \rangle_{n=0}^\infty$ a Sturmian
  sequence over a binary alphabet and $\beta$ an algebraic number with
  $|\beta|>1$.  We show that every such number is transcendental.
  More generally, for a given base~$\beta$ and given irrational
  number~$\theta$ we characterise the
  $\overline{\mathbb{Q}}$-linear independence of sets of the form
  $\left\{ 1,
  S_\beta(\boldsymbol{u}^{(1)}),\ldots,S_\beta(\boldsymbol{u}^{(k)})
  \right\}$, where $\boldsymbol{u}^{(1)},\ldots,\boldsymbol{u}^{(k)}$ are
  Sturmian sequences having slope $\theta$.
 
  We give an application of our main result to the theory of dynamical
  systems, showing  that for a contracted rotation on the unit circle
  with algebraic slope, its limit set is either finite or consists
  exclusively of transcendental elements other than its endpoints $0$
  and $1$.  This confirms a conjecture of Bugeaud, Kim, Laurent, and
  Nogueira~\cite{BKLN21}.
  \end{abstract}

  \section{Introduction}
  \label{sec:intro}
  A famous conjecture of Hartmanis and Stearns asserts that a real
  number $\alpha$ whose sequence of digits can be produced by a
  linear-time Turing machine (in the sense that for all $n$, given
  input $n$ in unary the machine outputs the first $n$ digits of
  $\alpha$ in time $O(n)$) is either rational or transcendental.  This
  conjecture remains open and is considered to be very difficult.  A
  weaker version---proposed by Cobham and eventually
  proved by Adamczewski, Bugeaud, and
  Luca~\cite{ABL}---asserts the transcendence of an irrational
  automatic real number.  The underlying intuition is that the
  sequence of digits of an irrational algebraic number cannot be too
  simple.  Indeed, the main technical result of~\cite{ABL} is that
  over an integer base every number whose sequence of digits has
  linear subword complexity is either rational or transcendental.
  Cobham's conjecture is an immediate corollary, given that automatic
  sequences have linear subword complexity.

  In this paper we prove a transcendence result for numbers whose
  digit sequences are Sturmian words (sometimes called mechanical
  words).  Such words have minimal subword complexity among
  non-ultimately periodic words and have a natural characterisation in
  terms of dynamical systems as codings of rotations on the unit
  circle.  The novelty of this work is that we handle expansions over
  an arbitrary algebraic base rather than just an integer base.  Here
  we are motivated by applications to control theory and dynamical
  systems.

  An infinite sequence $\boldsymbol{u}=u_0u_1u_2\ldots$ over a binary
  alphabet is said to be \emph{Sturmian} if the number $p(n)$ of
  different length-$n$ factors in $\boldsymbol u$ satisfies $p(n)=n+1$
  for all $n \in \mathbb{N}$, see~\cite{MH2}.  Coven and
  Hedlund~\cite{CH73} show that an infinite word such that
  $p(n)\leq n$ for some $n$ is necessarily ultimately periodic. Thus
  Sturmian words have minimal subword complexity among non-ultimately
  periodic words over a binary alphabet $\{0,1\}$.  The letters in a
  Sturmian word have a limiting frequency---the limit frequency of the
  letter $1$ is called the \emph{slope} of the word.  Related to this,
  Sturmian words have a natural characterisation in terms of dynamical
  systems, namely as codings of the orbits of irrational rotations on
  $\mathbb{R} / \mathbb{Z}$.  Perhaps the best known example of a
  Sturmian word is the \emph{Fibonacci word}.  This is defined as the
  limit $\boldsymbol{f}_\infty$ of the sequence
  $(\boldsymbol f_n)_{n=0}^\infty$ of finite strings over the binary
  alphabet $\{0,1\}$, defined by the recurrence $\boldsymbol f_0:=0$,
  $\boldsymbol f_1:=01$, and
  $\boldsymbol f_n=\boldsymbol f_{n-1}\boldsymbol f_{n-2}$ for all
  $n \geq 2$.  The limit is well defined since $\boldsymbol f_n$ is a
  prefix of $\boldsymbol f_{n+1}$ for all $n \in \mathbb{N}$.  The
  Fibonacci word has slope $1/\phi$, where $\phi=\frac{1+\sqrt{5}}{2}$ is the golden ratio.
  It so happens that the Fibonacci word is morphic, although it is not
  automatic.

  Let $\boldsymbol{u}$ be a Sturmian word over a finite alphabet
  $\Sigma \subseteq \overline{\mathbb{Q}}$ and let
  $\beta \in \overline{\mathbb{Q}}$ be such that $|\beta|>1$.  Then we
  call
  $S_\beta(\boldsymbol{u}):=\sum_{n=0}^\infty \frac{u_n}{\beta^n}$ a
  \emph{Sturmian number} with \emph{sequence of digits}
  $\boldsymbol{u}$ and \emph{base} $\beta$.\footnote{Our notion of
    Sturmian number is more permissive than that of Morse and
    Hedland~\cite{MH} who restricted to the case of an integer base
    $b>1$ and digit sequence $\boldsymbol{u}$ over alphabet
    $\{0,\ldots,b-1\}$.}  Ferenczi and Mauduit~\cite{FM} proved the
  transcendence of every number $S_\beta(\boldsymbol{u})$ over an
  integer base $\beta > 1$.  Their proof combined combinatorial
  properties of Sturmian sequences with a $p$-adic version of the
  Thue-Siegel-Roth Theorem, due to Ridout.  This result was
  strengthened by Bugeaud \emph{et al.}~\cite{BKLN21} to show
  $\overline{\mathbb{Q}}$-linear independence of sets of the form
  $\left\{1,S_\beta(\boldsymbol{u}^{(1)}),S_\beta(\boldsymbol{u}^{(2)})\right\}$ where
  $\boldsymbol{u}^{(1)},\boldsymbol{u}^{(2)}$ are Sturmian words having the same
  slope and $\beta>1$ is an integer.  In the case of an algebraic
  base $\beta$, Laurent and Nogueria~\cite{LN} observe that if
  $\boldsymbol u$ is a characteristic Sturmian word
  (cf. Section~\ref{sec:stuttering}), then the transcendence of
  $S_\beta(\boldsymbol{u})$ follows from a result of Loxton and Van
  der Poorten~\cite[Theorem 7]{LP} concerning transcendence of
  Hecke-Mahler series.

  In this paper we give a common generalisation of the above three
  results.  For every algebraic base $\beta$ and irrational slope
  $\theta$ we give sufficient and necessary conditions for
  $\overline{\mathbb{Q}}$-linear independence of a set of Sturmian numbers
  $\left\{ 1,
  S_\beta(\boldsymbol{u}^{(1)}),\ldots,S_\beta(\boldsymbol{u}^{(k)})
  \right\}$, where $\boldsymbol{u}^{(1)},\ldots,\boldsymbol{u}^{(k)}$, where
  are Sturmian sequences of
  slope $\theta$.  Our characterisation relies on a new combinatorial
  criterion on a sequence $\boldsymbol u$ that ensures transcendence
  of $S_\beta(\boldsymbol{u})$ for $\beta$ an algebraic base.  Similar
  to~\cite{BKLN21}, the Subspace Theorem plays a major role in our
  argument.  In~\cite{LOW} we give a more elaborate and powerful
  transcendence criterion that allows proving
  $\overline{\mathbb Q}$-linear independence results about Sturmian
  numbers (again with a common slope) over different algebraic bases.

  For a sequence $\boldsymbol u$ with linear subword complexity (i.e.,
  such that $\lim\inf_n \frac{p(n)}{n} < \infty$), it is shown
  in~\cite{AB1} that $S_\beta(\boldsymbol{u})$ is transcendental under
  the condition that $\beta$ is a Pisot number (i.e., a real algebraic
  integer greater than one all of whose Galois conjugates have
  absolute value less than one).  Compared to the main result of this
  paper, the class of sequences considered by~\cite{AB1} is more
  general (requiring merely linear subword complexity rather than the
  stronger condition of being Sturmian), but the condition on the base
  is more restrictive (being a Pisot number rather than merely an
  algebraic number of absolute value strictly greater than one).

  In Section~\ref{sec:contracted} we give an application of our main
  result to the theory of dynamical systems.  We consider the set $C$
  of limit points of a contracted rotation $f$ on the unit interval,
  where $f$ is assumed to have an algebraic contraction factor.  The
  set $C$ is finite if $f$ has a periodic orbit and is otherwise a
  Cantor set, that is, it is homeomorphic to the Cantor ternary set
  (equivalently, it is compact, nowhere dense, and has no isolated
  points).  In the latter case we show that all elements of $C$ except
  its endpoints $0$ and $1$ are transcendental.  Our result confirms a
  conjecture of Bugeaud, Kim, Laurent, and Nogueira, who proved a
  special case of this result in~\cite{BKLN21}.  We remark that it is
  a longstanding open question whether the actual Cantor ternary set
  contains any algebraic elements other than $0$ or $1$.

\section{Preliminaries}
\label{sec:prelim}
Let $K$ be a number field of degree $d$ and let $M(K)$ be the set of
\emph{places} of $K$.  We divide $M(K)$ into the collection of
\emph{infinite places}, which are determined either by an embedding of
$K$ in $\mathbb{R}$ or a complex-conjugate pair of embeddings of $K$
in $\mathbb{C}$, and the set of \emph{finite places}, which are
determined by prime ideals in the ring $\mathcal{O}_K$ of integers of $K$.

For $x \in K$ and $v \in M(K)$, define the
absolute value $|x|_v$ as follows: $|x|_v := |\sigma(x)|^{1/d}$ in
case $v$ corresponds to a real embedding
$\sigma:K\rightarrow \mathbb{R}$; $|x|_v := |\sigma(x)|^{2/d}$ in case
$v$ corresponds to a complex-conjugate pair of embeddings
$\sigma,\overline{\sigma}:K \rightarrow \mathbb{C}$; finally,
$|x|_v := N(\mathfrak{p})^{-\mathrm{ord}_{\mathfrak{p}}(x)/d}$ if $v$
corresponds to a prime ideal $\mathfrak{p}$ in $\mathcal{O}$
and $\mathrm{ord}_{\mathfrak{p}}(x)$ is the order
of $\mathfrak{p}$ as a divisor of the ideal $x\mathcal{O}$.
With the
above definitions we have the \emph{product formula}:
$\prod_{v \in M(K)} |x|_v = 1$ for all $x \in K^\ast$.  Given a set of
places $S\subseteq M(K)$, the ring $\mathcal{O}_S$ of
\emph{$S$-integers} is the subring comprising all $x \in K$ such
$|x|_v \leq 1$ for all finite places $v\in S$.

For $m\geq 2$ the \emph{absolute Weil height} of
$\boldsymbol{x}=(x_1,\ldots,x_m) \in K^m$ is defined to be
\[  H(\boldsymbol{x}):=\prod_{v \in M(K)}\max(|x_1|_v,\ldots,|x_m|_v)
  \, .\]
This definition is independent of the choice of field $K$ containing
$x_1,\ldots,x_m$.  Note the restriction $m\geq 2$ in the above
definition.  For $x \in K$ we define its height $H(x)$ to be $H(1,x)$.
For a non-zero polynomial $f = \sum_{i=0}^s a_i X^i \in K[X]$, where
$s\geq 1$, we define its height $H(f)$ to be the height of its
coefficient vector $(a_0,\ldots,a_s)$.

The following classical result of  Schlickewei will be instrumental in our approach.
\begin{theorem}[Subspace Theorem]
  Let $S \subseteq M(K)$ be a finite set of places, containing all
  infinite places and let $m\geq 2$.  For every $v\in S$ let
  $L_{1,v},\ldots,L_{m,v}$ be linearly independent linear forms in $m$
  variables with algebraic coefficients.  Then for any $\varepsilon>0$
  the solutions $\boldsymbol{x} \in \mathcal{O}_S^m$ of the inequality
\[ \prod_{v\in S} \prod_{i=1}^m |L_{i,v}(\boldsymbol{x}) |_v \leq
  H(\boldsymbol{x})^{-\varepsilon} \] are contained in finitely many
proper subspaces of $K^m$.
\label{thm:SUBSPACE}
\end{theorem}

We will also need the following more elementary proposition.  
\begin{proposition}{\cite[Proposition 2.3]{LEN97}}
  Let $f \in K[X]$ be a polynomial with at most 
  $k+1$ terms.  Assume that $f$ can be written as the sum of two 
  polynomials $g$ and $h$, where every monomial of $g$ has degree at 
  most $d_0$ and every monomial of $h$ has degree at least $d_1$. 
  Let $\beta$ be a root of $f$ that is not a root of unity.  If 
  $d_1-d_0> \frac{\log (k \, H(f))}{\log H(\beta) }$ then $\beta$ is a 
  common root of $g$ and $h$. 
\label{prop:gap}
  \end{proposition}

\section{Stuttering Sequences}
\label{sec:stuttering}
Let $A\subseteq\overline{\mathbb{Q}}$ be a finite alphabet.  An infinite
sequence $\boldsymbol{u}=u_0u_1u_2\ldots \in A^\omega$ is said to be
\emph{stuttering} if for all $w>0$ there exist sequences
$\langle r_n\rangle_{n=0}^\infty$ and
$\langle s_n\rangle_{n=0}^\infty$ of positive integers and $d\geq 2$
such that:
\begin{enumerate}
\item[S1] $\langle r_n\rangle_{n=0}^\infty$ is unbounded and
  $s_n \geq wr_n$ for all $n\in\mathbb N$;
\item[S2] for all $n\in \mathbb{N}$ there exist integers
  $0 \leq i_1(n)<\ldots < i_{d}(n) \leq s_n$ such that the strings
  $u_0\ldots u_{s_n}$ and $u_{r_n} \ldots u_{r_n+s_n}$ differ at the
  set of indices $\bigcup_{j=1}^d \{i_j(n),i_{j}(n)+1\}$;
\item[S3] we have %the \emph{expanding gaps properties:}
  $i_d(n)-i_1(n) = \omega(\log r_n)$ and, writing $i_0(n):=0$ and
  $i_{d+1}(n):=s_n$ for all $n$, we have $i_{j+1}(n)-i_j(n)=\omega(1)$
  for all $j\in\{0,1,\ldots,d\}$;
\item[S4] for all $n\in\mathbb{N}$ and $j\in\{1,2\ldots,d\}$ we have
  $u_{i_j(n)}  + u_{i_j(n)+1}=
  u_{i_j(n)+r_n}+u_{i_j(n)+r_n+1}$.
\end{enumerate}

The notion of a stuttering sequence is reminiscent of the
transcendence conditions of~\cite{AB1,BKLN21,FM} in that it concerns
periodicity in an infinite word.  Roughly speaking, a sequence
$\boldsymbol u$ is stuttering if for all $w>0$ there are arbitrarily
long prefixes of $\boldsymbol u$ that, modulo a fixed number of
mismatches, comprise $w$ repetitions of some finite word.  The fact
that the number $w$ of repetitions is arbitrary is key to our being
able to prove transcendence results over an arbitrary algebraic base
$\beta$.  In compensation, our condition allows repetitions with a
certain number of discrepancies.  This should be contrasted with the
notion of stammering sequence in~\cite[Section 4]{AB1}, where there is
no allowance for such discrepancies and in which the quantity
corresponding to $w$ is fixed.

\begin{example}
  To illustrate the notion of stuttering sequence, we recall the
  example of the Fibonacci word.  That this sequence is stuttering is
  a consequence of Theorem~\ref{thm:main1}.  Here in fact the sequence
  of shifts $\langle r_n \rangle_{n=0}^\infty$ witnessing that the
  Fibonacci word is stuttering is the Fibonacci sequence
  $\langle 1,1,2,3,5,\ldots\rangle$.  Below we align the Fibonacci
  word $\boldsymbol{f}_\infty$ with its shift
  $\boldsymbol{f}_\infty^{(5)}$ by $r_5=5$, underlining the mismatches
  which arise in consecutive pairs that satisfy Condition S4.
  \begin{align*}
 \boldsymbol{f}_\infty := \,  & 010010\underline{10}01001010010\underline{10}010010\underline{10}0100101001
      \ldots \\
 \boldsymbol{f}_\infty^{(5)} := \,   & 010010\underline{01}01001010010\underline{01}010010\underline{01}0100101001 \ldots
\end{align*}
      \end{example}

In what follows, we use the following representation of Sturmian
words.  Write $I:=[0,1)$ for the unit interval and given
$x \in \mathbb{R}$ denote the integer part of $x$ by $\lfloor x \rfloor$ and
the fractional part of $x$ by
$\{x\}:=x-\lfloor x \rfloor \in I$.  Let $0<\theta<1$ be an irrational number
and define the \emph{rotation map} $T = T_\theta :I \rightarrow I$ by $T(y)= \{ y+\theta \}$.  Given
$x \in I$, the \emph{$\theta$-coding of $x$} is the infinite sequence
$\boldsymbol{u}=u_1u_2u_3 \ldots$ defined by $u_n := 1$ if
$T^n(x) \in [0,\theta)$ and $u_n:=0$ otherwise.  As shown by Morse and
Hedlund, $\boldsymbol u$ is a Sturmian word and, up to changing at
most two letters, all Sturmian words over a binary alphabet arise as
codings of the above type for some choice of $\theta$ and $x$.  In
particular, for the purposes of establishing our transcendence results
we may work exclusively with codings as defined above.  The number
$\theta$ is equal to the {slope} of the Sturmian word, as defined in
Section~\ref{sec:intro}.
The $\theta$-coding of $0$ is in particular
called the \emph{characteristic Sturmian word of slope $\theta$}.

The main result of this section is as follows:
\begin{theorem}
  Let $\theta \in (0,1)$ be irrational.  Given a positive integer $k$,
  let $c_0,\ldots,c_{k} \in \mathbb{C}$ and $x_1,\ldots,x_k \in I$.
  Suppose that $x_i-x_j \not\in\mathbb{Z}\theta+\mathbb{Z}$ for all
  $i\neq j$.  Writing $\langle u_n^{(i)} \rangle_{n=0}^\infty$ for the
  $\theta$-coding of $x_i$, for $i=1,\ldots,k$, define
  $u_n:=c_0+\sum_{i=1}^kc_i u_n^{(i)}$ for all $n\in \mathbb{N}$.
  Then $\boldsymbol{u}=\langle u_n\rangle_{n=0}^\infty$ is stuttering.
        \label{thm:main1}
  \end{theorem}
  \begin{proof}
    We start by recalling some basic facts about the
    continued-fractions.
   Write $[a_0,a_1,a_2,a_3,\ldots]$ for the simple continued-fraction
expansion of $\theta$.  Given $n\in \mathbb{N}$, we write
$\frac{p_n}{q_n}:=[a_0,a_1,\ldots,a_n]$ for the $n$-th convergent.
Then $\langle q_n \rangle_{n=0}^\infty$ is a strictly increasing
sequence of positive integers such that
$\| q_n\theta \| = |q_n\theta - p_n|$, where $\| \alpha \|$ denotes
the distance of a given number $\alpha \in \mathbb{R}$ to the nearest
integer.  We moreover have that $q_n\theta-p_n$ and
$q_{n+1}\theta-p_{n+1}$ have opposite signs for all $n$.  Finally we
have the \emph{law of best approximation}: $q\in\mathbb{N}$ occurs as
one of the $q_n$ just in case $\|q\theta\|<\|q'\theta\|$ for all $q'$
with $0<q'<q$.

To establish that $\boldsymbol u$ is stuttering, given $w>0$ we define
$\langle r_n\rangle_{n=0}^\infty$ to be the subsequence of
$\langle q_n \rangle_{n=0}^\infty$ comprising all terms $q_n$ such
that $\|q_n\theta\| = q_n\theta-p_n > 0$.  Note that we either have
$r_n=q_{2n}$ for all $n$ or $r_n=q_{2n+1}$ for all $n$, so
$\langle r_n\rangle_{n=0}^\infty$ is an infinite sequence that
diverges to infinity.  Next, write $d=(k+1)w$ and for all
$n\in\mathbb{N}$ define $s_n$ be the greatest number such that the
words $u_0 \ldots u_{s_n}$ and $u_{r_n} \cdots u_{r_n+s_n}$ have
Hamming distance at most $2d$.  Since $\boldsymbol{u}$ is not
ultimately periodic, $s_n$ is thereby well-defined.

\textbf{Condition S2.}
Denote  the set of positions at which $u_0\ldots u_{s_n}$ and 
$u_{r_n} \ldots u_{s_n+r_n}$ differ by
\begin{gather}
  \Delta_n:=\big\{ m \in \{0,\ldots,s_n\} : u_m \neq u_{m+r_n}
  \big\} \, .
  \label{eq:Delta}
\end{gather}
We claim that for $n$ sufficiently large, $m\in\Delta_n$ if and only if
there exists $\ell \in\{1,\ldots,k\}$ such that
one of the following two conditions holds:
      \begin{enumerate}
\item[(i)]$T^m(x_\ell) \in [1-\|r_n\theta\|,1)$, 
\item[(ii)] $T^m(x_\ell) \in [\theta-\|r_n\theta\|,\theta)$. 
\end{enumerate}
We claim furthermore that for all $m$ there is most $\ell$ such that one of above
conditions holds.

Assuming the claim, since $T^m(x_\ell) \in [1-\|r_n\theta\|,1)$ if and
only if $T^{m+1}(x_\ell) \in [\theta-\|r_n\theta\|,\theta)$, it
follows that the elements of $\Delta_n$ come in consecutive pairs,
i.e., we can write
\[ \Delta_n=\bigcup_{j=1}^d \{i_j(n),i_j(n)+1\} \, , \]
where  $i_1(n)<\ldots<i_d(n)$ are the elements $m\in\Delta_n$ that 
satisfy Condition~(i) above for some $\ell$.

It remains to prove the claim.  To this end note that for a
fixed $\ell \in\{1,\ldots,k\}$ we have
$u_m^{(\ell)} \neq u_{m+r_n}^{(\ell)}$ iff exactly one of
$T^m(x_\ell)$ and $T^{m+r_n}(x_\ell)$ lies in the interval
$[0,\theta)$ iff either Condition~(i) or Condition~(ii) holds.  
Moreover, since $x_{\ell}-x_{\ell'} \neq \theta \pmod 1$ for 
$\ell\neq\ell'$, we see that for  $n$ sufficiently large 
there is at most one $\ell\in\{1,\ldots,k\}$ such that one of these
two conditions holds.  Equivalently, for all $m$
there is at most one $\ell$ such 
that $u_m^{(\ell)} \neq u_{m+r_n}^{(\ell)}$.    We deduce that
$u_m \neq u_{m+r_n}$ if and only if 
$u_m^{(\ell)} \neq u_{m+r_n}^{(\ell)}$ for some 
$\ell\in \{1,\ldots,k\}$.  This concludes the proof of the claim.

\textbf{Condition S1.}  Our objective is to show that $s_n \geq wr_n$
for all $n\in \mathbb{N}$.  We have already established that there are
$d=(k+1)w$ distinct $m \in \Delta_n$ that satisfy Condition~(i), 
above,  for some $\ell \in \{1,\ldots,k\}$.  Thus there exists
$\ell_0 \in \{1,\ldots,k\}$ and $\Delta'_n \subseteq \Delta_n$ such
that $|\Delta'_n| \geq w$ and all $m \in \Delta_n'$ satisfy
Condition~(i) for $\ell=\ell_0$.  In this case we have
$\| (m_1-m_2)\theta \| < \| r_n\theta\|$ for all
$m_1,m_2 \in \Delta_n'$.  By the law of best approximation it follows
that every two distinct elements of $\Delta_n'$ have difference
strictly greater than $r_n$.  But this contradicts $|\Delta_n'| = w$
given that $\Delta_n' \subseteq \{0,1,\ldots,wr_n\}$.\qed

    \textbf{Condition S3.}
By definition of $i_1(n),\ldots,i_d(n)$, for all
  $j\in\{1,\ldots,d\}$ there exists $\ell_j(n) \in \{1,\ldots,k\}$ with 
  $T^{i_j(n)}(x_{\ell_j(n)}) \in [1-\|r_n\theta\|,1)$.
Now, for all $n\in\mathbb{N}$ and $1\leq j_1<j_2\leq d$ we have
    \begin{gather}
\|(i_{j_2}(n)-i_{j_1}(n))\theta+x_{\ell_{j_2}(n)}-x_{\ell_{j_1}(n)}\| \leq 
\|r_n\theta\| \, .
    \label{eq:KEY}
    \end{gather}
    We claim that the left-hand side of~\eqref{eq:KEY} is non-zero.
    Indeed, the claim holds if $\ell_{j_2}(n)=\ell_{j_1}(n)$ because
    $\theta$ is irrational, while the claim also holds in case
    $\ell_{j_2}(n)\neq \ell_{j_1}(n)$ since in this case we have
    $x_{\ell_{j_2}(n)}-x_{\ell_{j_1}(n)}
    \not\in\mathbb{Z}\theta+\mathbb{Z}$ by assumption.
    Since moreover the right-hand side of~\eqref{eq:KEY} tends to zero
    as $n$ tends to infinity, we have that
    $i_{j_2}(n) - i_{j_1}(n)=\omega(1)$.  On the other hand, if
    $\ell_{j_2}(n)=\ell_{j_1}(n)$ then we even have
    $i_{j_2}(n)-i_{j_1}(n) \geq r_n = \omega(\log r_n)$ by the law of best
    approximation.  Hence we certainly have $i_d(n)-i_1(n) =
    \omega(\log r_n)$.
    
    Finally, defining $i_0(n):=0$ we have $i_1(n)-i_0(n)=\omega(1)$ by
    the requirement that
    $T^{i_1(n)}(x_{\ell_1(n)}) \in [1-\|r_n\theta\|,1)$ and the fact
    that $\|r_n\theta\|$ converges to $0$.  Setting $i_{d+1}(n):=s_n$
    for all $n$, we also have $i_{d+1}(n)-i_d(n)=\omega(1)$ by the
    maximality condition in the definition of $s_n$.

\textbf{Condition S4.}
Consider $m \in \Delta_n$ satisfying Condition~(i)
above, i.e., such that $T^m(x_\ell) \in [1-\|r_n\theta\|,1)$ for some
$\ell \in \{1,\ldots,k\}$.  Then we have
\[ u^{(\ell)}_m=0,\, u^{(\ell)}_{m+1}=1 \quad\text{and}\quad 
    u^{(\ell)}_{m+r_n}=1,\, u^{(\ell)}_{m+r_n+1}=0 \, . \]
Moreover for all $\ell' \neq \ell$ and $n$ sufficiently large we have
  \[ u^{(\ell')}_m=u^{(\ell')}_{m+r_n} \quad\text{and}\quad 
    u^{(\ell')}_{m+1}= u^{(\ell')}_{m+r_n+1} \, . \]
We conclude that $u_m+u_{m+1}=u_{m+r_m}+u_{m+r_n+1}$, establishing
Condition S4.
\end{proof}
\section{A Transcendence Result}

\begin{theorem}
  Let $A$ be a finite set of algebraic numbers and suppose that
  $\boldsymbol{u} \in A^\omega$ is a stuttering sequence. Then for any
  algebraic number $\beta$ with $|\beta|>1$, the sum
  $\alpha:=\sum_{n=0}^\infty \frac{u_n}{\beta^n}$ is transcendental.
\label{thm:main}
\end{theorem}
  \begin{proof}
    Suppose for a contradiction that $\alpha$ is algebraic.  By
    scaling we can assume without loss of generality that $A$ consists
    solely of algebraic integers.  Let $K=\mathbb{Q}(\beta)$ be the
    field generated over $\mathbb{Q}$ by $\beta$ and write
    $S\subseteq M(K)$ for the set comprising all infinite places of
    $K$ and all finite places of $K$ corresponding to prime-ideal
    divisors of the ideal $\beta\mathcal{O}_K$.
    
    Applying the stuttering condition (for a value of $w$ to be
    determined later), we obtain $d\geq 2$ such that for all
    $n\in\mathbb N$ there are positive integers
    $r_n,s_n,i_1(n),\ldots,i_{d}(n)$ satisfying conditions S1--S4. By condition S2, for all $n$ if we define
    \[ c_{j}(n):= (u_{i_j(n)+r_n}-u_{i_j(n)})+ (u_{i_{j}(n)+r_n+1} -
      u_{i_{j}(n)+1})\beta^{-1}, \quad j\in\{1,2,\ldots,d\}
    \]
    and $\alpha_n := \sum_{j=0}^{r_n} {u_j}{\beta^{r_n-j}}$
    then we have
    \begin{gather}
\left|\beta^{r_n}\alpha - \alpha - \alpha_n - c_1(n) \beta^{-i_1(n)} -  \cdots - c_d(n)
  \beta^{-i_d(n)} \right| < |\beta|^{-s_n} \, , 
\label{eq:ONE}
\end{gather}
Note that $c_1(n),\ldots,c_d(n)$ are non-zero by Condition S4.  By
passing to a subsequence we can furthermore assume without loss of
generality that $c_1=c_1(n),\ldots,c_d=c_d(n)$ are constant,
independent of $n$.

To set up the application of the Subspace Theorem, 
define a family of linear forms $L_{i,v}$, for $1\leq i \leq 3+d$ and
$v\in{S}$, by
\[\begin{array}{rcl}
L_{i,v}(x_1,\ldots,x_{3+d}) &:= & x_i \;\mbox{for all $ (i,v)\neq
                                  (3,v_0)$, and }\\
  L_{3,v_0}(x_1,\ldots,x_{3+d}) &:= & \alpha x_1-\alpha x_2- x_3 -
                                      \sum_{j=1}^{d} c_j x_{3+j} \,  .
\end{array}\]
Write
$\boldsymbol{b}_n:=\left(\beta^{r_n},1,\alpha_n,\beta^{-i_1(n)},\ldots,\beta^{-i_{d}(n)}\right)$ and
let $M\geq 2$ be an upper bound of the set of real numbers
\[ \left\{ |\gamma |_v : \gamma \in \{\beta\} \cup A ,\, v \in S\right\} \, . \]
Then for all $v\neq v_0$ we have
\[ |L_{3,v}(\boldsymbol{b}_n)|_v = |\alpha_n|_v \leq
    \sum_{j=0}^{r_n} M^{j+1} \leq M^{r_n+2}  \, , \] 
while
$|L_{3,v_0}(\boldsymbol{b}_n)|_{v_0} \leq
|\beta|^{-s_n/{\deg(\beta)}}$  by~\eqref{eq:ONE}.  Furthermore, for
$i\neq 3$, by the product formula we have $\prod_{v\in {S}}
|L_{i,v}(\boldsymbol{b}_n)|_v = 1$.  Altogether we have
\begin{gather}
\prod_{v \in {S}} \prod_{i=1}^{d+3} |
L_{i,v}(\boldsymbol{b}_n) |_v \leq M^{(r_n+2)|S|}
\cdot |\beta|^{-s_n/{\deg(\beta)}} \,  .
\label{eq:INEQ3}
\end{gather}

Since $s_n\geq wr_n$ we have that for $w$ sufficiently large the right-hand side
of~\eqref{eq:INEQ3} is less than $|\beta|^{-s_n/{2\deg(\beta)}}$.  On the other hand
there exists a constant $c$ such that the height of $\boldsymbol b_n$
satisfies the bound $H(\boldsymbol{b}_n) \leq |\beta|^{c s_n}$ for all
$n$.  Thus there exists $\varepsilon>0$ such that the right-hand side
of~\eqref{eq:INEQ3} is at most $H(\boldsymbol{b}_n)^{-\varepsilon}$
for all $n$.  Since $\boldsymbol{b}_n$ is a vector of $S$-units we can
apply the Subspace Theorem to obtain a non-zero linear form
$L(x_1,\ldots,x_{3+d})$ with coefficients in $K$ such that
$L(\boldsymbol{b}_n)=0$ for infinitely many $n\in\mathbb{N}$.

Denote by $\mathrm{vars}(L) \subseteq \{x_1,\ldots,x_{3+d}\}$ the set
of variables that appear in $L$ with non-zero coefficient.  We claim
that $x_3 \in \mathrm{vars}(L)$.  Indeed, suppose for a contradiction
that $x_3 \not\in \mathrm{vars}(L)$. Then for all $n$,
$L(\boldsymbol{b}_n)$ is a fixed linear combination of the numbers
$\beta^{r_n},1,\beta^{-i_1(n)},\ldots,\beta^{-i_d(n)}$.  By Item~S3
the gaps beween successive exponents in these powers of $\beta$ tend
to infinity with $n$ and hence a fixed linear combination of such
powers cannot vanish for arbitrarily large $n$.

We have that $L(\boldsymbol{b}_n)$ is a linear combination of a most 
$r_n+d+1$ powers of $\beta$, whose respective exponents lie in the set 
$\{0,1,\ldots,r_n\} \cup \{-i_1(n),\ldots,-i_d(n)\}$. 
From Item S3
there exists $j_0 \in \{1,\ldots,d-1\}$ such that
$i_{j_0+1}(n)-i_{j_0}(n) = \omega(\log r_n)$.  By
Proposition~\ref{prop:gap} the condition $L(\boldsymbol{b}_n)=0$
entails, for $n$ sufficiently large, that $\mathrm{vars}(L)$ is
contained either in $\{x_1,\ldots,x_{j_0+3}\}$ or in
$\{x_{j_0+4},\ldots,x_d\}$.  Since we know that $x_3 \in
\mathrm{vars}(L)$ the former inclusion applies.

We have established that
$x_3 \in \mathrm{vars}(L)\subseteq \{x_1,\ldots,x_{j_0+3}\}$.  Thus by
a suitable linear combination of the forms $L_{3,v_0}$ and $L$, so as
to eliminate the variable $x_3$, we obtain a non-zero linear form
$L'(x_1,\ldots,x_{3+d})$ with algebraic coefficients that does not
mention $x_3$ and such that $|L'(\boldsymbol{b}_n)|<|\beta|^{-s_n}$
for infinitely many $n$.  Note that $L'(\boldsymbol{b}_n)$ is a fixed
linear combination of at most $d+2$ powers of $\beta$, with respective
exponents in the set $\{r_n,0,-i_1(n),\ldots,-i_d(n)\}$.  Moreover by
Item S3 the gaps between consecutive elements of this set tend to
infinity with $n$.  It follows that
$|L'(\boldsymbol{b}_n)|\gg |\beta|^{-i_d(n)}$.  But since
$s_n-i_{d}(n) = \omega(1)$, this contradicts
$|L'(\boldsymbol{b}_n)|<|\beta|^{-s_n}$.
\end{proof}

We have the following immediate corollary of Theorem~\ref{thm:main1}
and Theorem~\ref{thm:main}.   
\begin{theorem}
  Let $\beta$ be an algebraic number with $|\beta|>1$.  Let 
  $0<\theta<1$ be irrational and let $x_1,\ldots,x_k \in I$ be 
  such that $x_i-x_j \not\in\mathbb{Z}\theta+\mathbb{Z}$ for 
  $i\neq j$.  For $i=1,\ldots,k$, define 
  $\alpha_i : = \sum_{n=0}^\infty \frac{u_n^{(i)}}{\beta^n}$, where 
  $\langle u_n^{(i)} \rangle_{n=0}^\infty$ is the $\theta$-coding of 
  $x_i$.  Then the set $\{1,\alpha_1,\ldots,\alpha_k\}$ is linearly 
  independent over the field $\overline{\mathbb{Q}}$ of algebraic numbers. 
  \label{thm:main2}
  \end{theorem}

\section{Application to Limit Sets of Contracted Rotations}
\label{sec:contracted}
Let $0<\lambda,\delta<1$ be real numbers such that $\lambda+\delta>1$.
We call the map $f= f_{\lambda,\delta} : I \rightarrow I$ given by
$f(x) := \{ \lambda x + \delta \}$ a \emph{contracted rotation} with
slope $\lambda$ and \emph{offset} $\delta$.  Associated with $f$ we
have the map
$F=F_{\lambda,\delta} : \mathbb{R}\rightarrow \mathbb{R}$, given by
$F(x)=\lambda \{x \} + \delta + \lfloor x \rfloor $.  We call $F$ a
\emph{lifting} of $f$: it is characterised by the properties that
$F(x+1)=F(x)+1$ and $\{ F(x) \} = f( \{x\} )$ for all
$x\in \mathbb{R}$.  The \emph{rotation number}
$\theta=\theta_{\lambda,\delta}$ of $f$ is defined by
\[ \theta:= \lim_{n\rightarrow \infty} \frac{F^n(x_0)}{n} \, , \]
where the limit exists and is independent of the initial point
$x_0\in\mathbb{R}$.

If the rotation number $\theta$ is irrational then the restriction of  
$f$ to the \emph{limit set} $\bigcap_{n\geq 0} f^n(I)$ is  
topologically conjugated to the rotation map  
$T = T_\theta:I\rightarrow I$ with $T(y) = \{y+\theta\}$.  The closure  
of the limit set is a Cantor set $C=C_{\lambda,\delta}$, that is, $C$ 
is compact, nowhere dense, and has no isolated points.  On the other  
hand, if $\theta$ is rational then the limit set $C$ is the unique  
periodic orbit of $f$.  For each choice of slope $0<\lambda<1$ and
irrational rotation number $0<\theta<1$, there exists a unique offset
$\delta$ such that $\delta+\lambda>1$ and the map $f$ has rotation
number $\theta$.  It is known that such $\delta$ must be
transcendental if $\lambda$ is algebraic~\cite{LN}.

\begin{figure}
\begin{center}
  \begin{tikzpicture}
\draw [thick] (0,0)--(5,0) ; 
\draw [thick] (0,0)--(0,5) ; 
\coordinate (Y0) at (0,5) ; 
\coordinate (Y1) at (0,4) ; 
\coordinate (Y2) at (0,1) ; 
\coordinate (M0) at (5,1) ; 
\coordinate  (X0) at (0,0) ; 
\coordinate (X1) at (2.5,0) ; 
\coordinate (X2) at (5,0) ; 
\coordinate (M1) at (2.5,5) ; 

\node at (Y0) [left=0.15] {$1$}; 
\node at (Y1) [left=0.15] {$\delta$}; 
\node at (Y2) [left=0.15] {$\delta+\lambda - 1$}; 
\node at (X0) [below=0.25, left=0.15] {$0$} ; 
\node at (X1) [below=0.15] {$\frac{1-\delta}{\lambda}$}; 
\node at (X2) [below=0.15] {$1$} ; 
\draw [dashed] (Y2)--(M0) ; 
\draw [dashed] (X1)--(M1) ; 
\draw [thick] (Y1)--(M1) ; 
\draw [thick] (X1)--(M0) ; 
\end{tikzpicture}
\caption{A plot of $f_{\lambda,\delta} : I \rightarrow I$}
\end{center}
\end{figure}
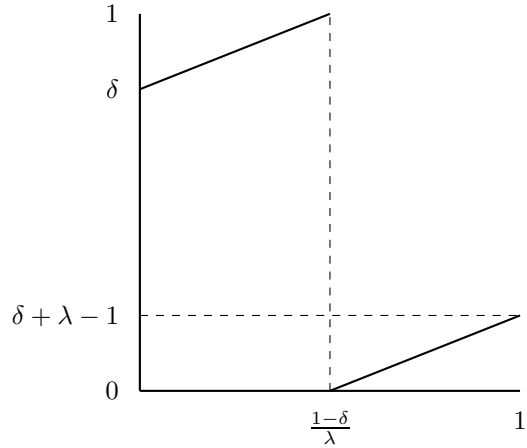

The main result of this section is as follows:

\begin{theorem}
\label{thm:cantor}
Let $0<\lambda,\theta<1$ be such
that $\lambda$ is algebraic and $\theta$ is irrational.  Let $\delta$
be the unique offset such that the contracted rotation
$f_{\lambda,\delta}$ has rotation number $\theta$.  Then every element
of the limit set $C_{\lambda,\delta}$ other than $0$ and $1$
is transcendental.
\end{theorem}

A special case of Theorem~\ref{thm:cantor}, in which $\lambda$ is
assumed to be the reciprocal of an integer, was proven
in~\cite[Theorem 1.2]{BKLN21}.  In their discussion of the latter
result the authors conjecture the truth of Theorem~\ref{thm:cantor},
i.e., the more general case in which $\lambda$ may be algebraic.  As
noted in~\cite{BKLN21}, while $C_{\lambda,\delta}$ is homeomorphic to
the Cantor ternary set, it is a longstanding open problem, formulated
by Mahler~\cite{Mahler84}, whether the Cantor ternary set
contains irrational algebraic elements.

\begin{proof}[Proof of Theorem~\ref{thm:cantor}]
  For a real number $0<x<1$ define
  \begin{eqnarray*}
    \xi_{x}&:= &\sum_{n \geq 1} \left( \lceil x + (n+1)\theta 
    \rceil - \lceil x + n\theta \rceil \right) \lambda^{n} \\
    \xi'_{x} &:= & \sum_{n\geq 1} \left( \lfloor x + (n+1)\theta 
    \rfloor - \lfloor x + n\theta \rfloor \right) \lambda^{n} \, .
\end{eqnarray*}
Note that for all $x$ the binary sequence 
$\langle \, \lceil x + (n+1)\theta \rceil - \lceil x + n\theta \rceil  
: n \in \mathbb{N} \, \rangle$ is  
the coding of $-x-\theta$ by $1-\theta$ (as defined in
Section~\ref{sec:stuttering}) and hence is Sturmian of
slope $1-\theta$.  Similarly, the binary sequence
$\langle \, \lfloor x + (n+1)\theta \rfloor - \lfloor x + n\theta \rfloor 
: n \in \mathbb{N} \, \rangle$ is the coding of $x+\theta$ by $\theta$
and hence is Sturmian of slope $\theta$.  Thus for all $x$, both
$\xi_x$ and $\xi'_x$ are Sturmian numbers.

  It is shown in~\cite[Lemma 4.2]{BKLN21}\footnote{The proof
    of the lemma is stated for $\beta$ an integer but carries over
    without change for $\beta$ algebraic.} that for every element of
  $y \in C_{\lambda,\delta} \setminus \{0,1\}$, either there exists
  $z \in \mathbb{Z}$ and $0<x<1$
with $x \not\in \mathbb{Z}\theta+\mathbb{Z}$
  such that
  \[ y = z + \xi_{0} - \xi_{-x} \] or else there exists
  a strictly positive integer $m$ and 
  $\gamma \in \mathbb{Q}(\beta)$ such
  that
  \[ y = \gamma + (1-\beta^{-m})  \, \xi'_{0} \, .\] 
In either case, transcendence of $y$
  follows from Theorem~\ref{thm:main2}.
\end{proof}

\textbf{Acknowledgements.}  The authors would like to thank Pavol
Kebis and Andrew Scoones for helpful feedback and corrections.

\end{document}